\documentclass[copyright,creativecommons]{eptcs}

\usepackage{booktabs} % For formal tables

\usepackage{amsfonts}
\usepackage{mathrsfs}
\usepackage{amssymb}
\usepackage{latexsym}
\usepackage{graphicx}
\usepackage{enumerate}
\usepackage{amsmath}
\usepackage{mathtools}
\usepackage{color}
\usepackage{float}
\usepackage{array}

\usepackage{multicol}					 % column environment for arranging figures

\mathtoolsset{centercolon}

\usepackage{amssymb,tikz}

\newcommand{\commentout}[1]{}

\newcommand{\defin}[1]{\textbf{#1}}

\newcommand{\lthen}{\rightarrow}

\renewcommand{\phi}{\varphi}
\renewcommand{\models}{\vDash}
\newcommand{\nmodels}{\nvDash}

\newcommand{\dimp}{\Leftrightarrow}

\usepackage{amsthm}
\newtheorem{theorem}{Theorem}

\newtheorem{proposition}[theorem]{Proposition}

\theoremstyle{definition}
\newtheorem{example}[theorem]{Example}

\title{Endogenizing Epistemic Actions}
\author{
Will Nalls
\institute{Carnegie Mellon University\\
      Pittsburgh, USA}
    \email{will.nalls@gmail.com}
\and
Adam Bjorndahl
\institute{Carnegie Mellon University\\
      Pittsburgh, USA}
	\email{abjorn@andrew.cmu.edu}
}
\def\titlerunning{Endogenizing Epistemic Actions}
\def\authorrunning{W.~Nalls \& A.~Bjorndahl}
\begin{document}

\maketitle

\begin{abstract}
Through a series of examples, we illustrate some important drawbacks that the \emph{action logic} framework suffers from in its ability to represent the dynamics of information updates. We argue that these problems stem from the fact that the \emph{action model}, a central construct designed to encode agents' uncertainty about actions, is itself effectively common knowledge amongst the agents. In response to these difficulties, we motivate and propose an alternative semantics that avoids them by (roughly speaking) endogenizing the action model. We discuss the relationshop to action logic, and provide a sound and complete axiomatization.
\end{abstract}

%\begin{CCSXML}
%<ccs2012>
%<concept>
%<concept_id>10003752.10003790.10003793</concept_id>
%<concept_desc>Theory of computation~Modal and temporal logics</concept_desc>
%<concept_significance>300</concept_significance>
%</concept>
%</ccs2012>
%\end{CCSXML}

%\ccsdesc[300]{Theory of computation~Modal and temporal logics}

%\keywords{Epistemic logic, action logic, information update}

%\thanks{}
%gets rid of copyright notice and the conference info
%\setcopyright{none}
%\acmConference{}{}{}

% The default list of authors is too long for headers}
%\renewcommand{\shortauthors}{G. Zhou et al.}

%introduction
\section{Introduction}

\emph{Action Logic} (AL) is a framework for reasoning about how knowledge and belief changes on the basis of incoming information \cite{BMS98,BM04,vDvdHK08}.\footnote{Terminology varies; some authors instead use \emph{Dynamic Epistemic Logic} to refer to this framework \cite[\S 2.2.2]{KP}. We follow van Ditmarsch et al.~\cite{vDvdHK08} in using it instead as an umbrella term for a collection of thematically related logics of information change, including Action Logic.}
Information is conveyed in the form of ``epistemic actions'', with a canonical example being \textit{public announcements} \cite{Plaza07}.
Unlike \emph{Public Announcement Logic} (PAL), however, AL does not presume that all epistemic actions are public in the sense of becoming common knowledge, nor that such actions can always be distinguished from one another by the agents. Uncertainty about epistemic actions is explicitly encoded in a structure called an \emph{action model}; this allows for the representation of scenarios in which agents can be uncertain about which action has taken place.

By design, this formalism is well-equipped to capture uncertainty about actions themselves; however, we argue in this paper that the AL framework is ill-suited to the representation of \textit{higher-order} uncertainty about actions.
%adam1
%In essence,
Roughly speaking,
this is because the action model that captures uncertainty about actions is itself effectively common knowledge amongst the agents,
%adam1
%and so AL struggles
making it awkward
to encode, for example, one agent's uncertainty about another agent's uncertainty about actions.

We expose this difficulty through a series of motivating examples. We demonstrate that AL can capture higher-order uncertainty about actions only by expanding the action model in such a way as to essentially ``pre-encode'' the desired uncertainty; this makes choosing an appropriate action model for any given application problematically post hoc. Furthermore, we show that in such cases small variations in the background epistemic conditions require corresponding alterations to the action models in order to ensure that the ``pre-encoded'' uncertainty maintains the right form. These observations seriously undermine the practical applicability of AL as a tool for reasoning about information updates.

In response to these challenges, we reformulate the semantics by
%adam1
%\textit{endogenizing}
``endogenizing''
the action model; somewhat more precisely, we allow each agent's uncertainty about actions to be world-dependent, and therefore itself subject to uncertainty. Revisiting our examples, we show that these revised semantics completely circumvent the earlier
%adam1
%difficulties. In fact, we show that
difficulties;
our semantics capture \textit{formally} the informal process underlying the aforementioned post hoc expansion of the action model.
%We discuss the relationship of this new semantics to AL, and provide a sound and complete axiomatization.

The idea of representing higher-order uncertainty about epistemic actions by encoding extra information about the agents (and their perceptions of such actions) into the state space is a natural one; it occurs also in \cite{BvDHLPS16}, in which Bolander et al.~study a more general class of announcements that may not be entirely public in that (loosely speaking) some agents may not be ``listening''. In essence, their semantics work by encoding into each possible world whether or not each agent is ``paying attention'' at that world. As the authors show, such ``attention-based'' announcements can also be described using action models. The present work is more general in that we begin with the full action model framework rather than PAL, and we encode into each world the full spectrum of each agent's uncertainty regarding epistemic actions, not just whether or not they are attentive.

The rest of the paper is organized as follows. In Section \ref{sec:act}, we motivate and define the basic AL framework, and present a series of examples intended to illustrate its limitations. In Section \ref{sec:adj}, we present our new semantics and show how it deals with the problems discussed in the previous section. Section 4 presents an axiomatization.

%actions logic
\section{Action Logic} \label{sec:act}

We begin by reviewing the foundational definitions and motivations of AL, largely following van Ditmarsch et al.~ \cite[Chapter 6]{vDvdHK08}.%
%
%
%
%
%adam: replaced with a quicker review
\commentout{
One line of investigation in formal epistemology takes epistemic scenarios to be captured by Kripke models. 
The particular Kripke models used here will be called \emph{epistemic models}:

\begin{mydef}
	An \defin{epistemic model} is a structure of the form $M = \langle W, \{\sim_j: j\in G\},V\rangle $, where:

\begin{itemize}
	\item $W$ is a set of possible states.
	\item $G$ is a set of agents.
	\item $\sim_j$ is an equivalence relation on $W$. $w_0 \sim_a w_1$ should be read as `agent a cannot distinguish $w_0$ from $w_1$'.
		We will use $[w_0]^a$ to denote the equivalence class of $w_0$ under $\sim_a$.
	\item $V$ is a function from $\mathbf{Prop}$, a countable set of propositions, to sets of states. $V(p)$ is the set of states in $W$ where $p$ is true.\footnote{Note that these propositions are intended to be ``facts about the world'', rather than ``facts about agents' epistemic states'' -- these will be captured by means of the relations.} 
		$V$ is extended to \eval{}, a valuation over the language $\mathcal{L}_{St}$.
\end{itemize}

\end{mydef}

We will call the language for the scenarios captured by epistemic model the language of epistemic logic,
	denoted $\mathcal{L}_{EL}$:
	
\[ p | \neg \phi | \phi \wedge \psi | K_j \phi \]

Epistemic formulas are evaluated at states in a model.
The semantics for propositions and boolean operators are as expected, and the semantics for $K_j \phi$ are given by:

\[ (M,w) \vDash K_j \phi \text{ just when } (\forall w'\in [w]^j)(M,w')\vDash \phi \]

$(M,w)\vDash K_j \phi$ may be read as `at world $w$, agent $j$ knows that $\phi$'.\footnote{In cases where the model is salient, we may write $w\vDash \phi$ instead of $(M,w) \vDash \phi$. Also, I will use $\hat{K_j}$ to represent the dual of $K_j$.}
The conditions on the indistinguishability relations -- that they be equivalence relations -- result in the knowledge operators having paricular properties. 
It is well-known that this class of models is axiomatized by the system \textbf{S5}.

Taking epistemic models to be representations of epistemic scenarios, a study of epistemic dynamics -- movements between epistemic scenarios -- will look at \emph{movements between epistemic models}.
Such movements may be the result of different \emph{epistemic actions}, or actions that cause changes in agents' epistemic states. 
The action of publically announcing a truth to all agents has been captured by \emph{Public Announcement Logic} (PAL).\footnote{This logic was originally proposed in \cite{plaza89}. A helpful exposition and summary of results may be found in \cite{vDvdHK08}. PAL has been axiomatized,
	and is reducible to epistemic logic.}

To illustrate this type of epistemic action, consider a scenario where colleagues Anne and Bob are discussing whether 
	a particular company policy passed in this morning's board meeting: both are ignorant of whether or not the policy passed. 
Taking $p$ to be the proposition `the policy passed', this scenario may be represented with the following epistemic model:\footnote{Since all epistemic models will be, unless stated otherwise, equivalence relations, reflexive loops and relations resulting from transitivity will be suppressed in the diagrams.}

\includegraphics[scale=.25 ]{diagrams/.eps}

$w_0$ and $w_1$ are the names of the possible worlds; $p$ and $\neg p$ are true at these worlds, respectively.
The edge labelled with $a$ and $b$ denotes that either agent cannot distinguish between these possible worlds. 
As desired, the model satisfies $M\vDash (\neg K_a p \wedge \neg K_b p) \wedge (\neg K_a \neg p \wedge \neg K_b \neg p)$.
Now consider the epistemic action characterized by a friend, Carl, coming by and truthfully announcing to the pair that $p$ is true. 
This removes, for both Anne and Bob, the possibility that $p$ is false, resulting in the model:

\includegraphics[scale=.25]{diagrams/d2.eps}

In the resulting model, both agents know that $p$ is true.
$PAL$ captures this sort of dynamic between epistemic models by means of \emph{model restriction}.
That is, in order to evaluate what is true after a public announcement of $\phi$ in a situation modeled by $M$, $PAL$ examines $M|_\phi$, the restriction of $M$ to all worlds where $\phi$ is true.
This restriction maintains all those relations between worlds in which $\phi$ is true -- this means that the only way for an agent to lose uncertainty is for a possibility to be removed for \emph{all} agents.
Clearly, this limits the types of epistemic actions which PAL is suited to model.

For instance, consider the epistemic action characterized by our friend Carl joining Bob and Anne, announcing that he knows whether or not $p$, and whispering to Bob that $p$ is true so that Anne cannot hear.
The result of this action may be captured by the following model:

\includegraphics[scale=1]{diagrams/d3.eps}

Clearly, this cannot be arrived at via node deletion.
Rather, the movement between models is one of \emph{edge deletion} -- the elimination of Bob's uncertainty with regard to $p$.
To account for this dynamic, a more general framework is required; Action Logic provides this framework.
}
%end comment-out
%
%
%
%
This logic is an extension of standard epistemic logic, so we begin there.

Fix a countable set of primitive propositions \textsc{prop} and a finite set of agents $G$. Let $\mathcal{L}_{EL}$ denote the language recursively defined as follows:
$$\phi ::= p \, | \, \neg \phi \, | \, \phi \land \psi \, | \, K_j \phi,$$
where $p \in \textsc{prop}$ and $j \in G$. We read $K_j \phi$ as ``agent $j$ knows that $\phi$''. Thus, $\mathcal{L}_{EL}$ is a language for reasoning about the knowledge of the agents in $G$. The other Boolean connectives can be defined in the usual way; we write $\hat{K}_{j}$ to abbreviate $\lnot K_{j} \lnot$, and read $\hat{K}_{j} \phi$ as ``agent $j$ considers it possible that $\phi$''.
 
An \defin{epistemic model} is a structure of the form $M = \langle W, \{\sim_j: j\in G\},V\rangle $, where:
\begin{itemize}
	\item $W$ is a (nonempty) set of \emph{states},
	\item $\sim_j$ is an equivalence relation on $W$,
	\item $V: \textsc{prop} \to 2^{W}$ is a \emph{valuation function}.
\end{itemize}
Intuitively, $V$ specifies for each primitive proposition those states where it is true, while the relations $\sim_{j}$ capture indistinguishability from the perspective of agent $j$. These intuitions are formalized in the following semantic clauses:
$$
\begin{array}{lcl}
(M,w) \models p & \textrm{ iff } & w \in V(p)\\
(M,w) \models \lnot \phi & \textrm{ iff } & (M,w) \nmodels \phi\\
(M,w) \models \phi \land \psi & \textrm{ iff } & (M,w) \models \phi \textrm{ and } (M,w) \models \psi\\
(M,w) \models K_{j} \phi & \textrm{ iff } & (\forall w' \in [w]^{j})((M,w) \models \phi),
\end{array}
$$
where $[w]^{j}$ denotes the equivalence class of $w$ under $\sim_j$. Thus, the Boolean connectives are interpreted as usual, and $K_{j}\phi$ is true at $w$ precisely when $\phi$ is true at all worlds that agent $j$ cannot distinguish from $w$. Insisting that the indistinguishability relations be equivalence relations results in a logic of knowledge that is \emph{factive} and \emph{fully introspective}. For a more thorough development of differing logics of knowledge we direct the reader to \cite{FHMV}.
%adam
%We consider weaker logics of knowledge in \draft{sec}.

On top of this basic epistemic framework, Action Logic adds a layer of structure that aims to capture the epistemic dynamics of information update.
An \defin{action model} is a structure of the form
$$A = \langle \Sigma, \{\approx_j : j\in G\}, Pre \rangle$$
where:
\begin{itemize}
	\item $\Sigma$ is a (nonempty) set of \emph{epistemic actions}, 
    \item $\approx_j$ is an equivalence relation on $\Sigma$,
    \item $Pre: \Sigma \to \mathcal{L}_{EL}$ is a \emph{precondition function}.
\end{itemize}
Intuitively, the relation $\approx_{j}$ captures indistinguishability of actions from the perspective of agent $j$, while the function $Pre$ captures the background conditions $Pre(\sigma)$ that must hold for a given action $\sigma$ to be successfully performed. In short, an action model specifies a set of epistemic actions that can be executed together with their preconditions and the extent to which they can be individuated by the agents.

To formalize these intuitions, we must define the process by which an epistemic model $M$ is updated based on the performance of an epistemic action from $A$.
This is captured in the \defin{updated model} 
$$M^A = \langle W_\Sigma , \{ \sim'_j : j\in G\}, V' \rangle$$
defined as follows:
\begin{itemize}
	\item $W_\Sigma = \{ (w,\sigma) \: : \: (M,w)\vDash Pre(\sigma) \}$,
	\item $(w_0,\sigma_0) \sim'_j (w_1,\sigma_1)$ iff $w_0 \sim_j w_1$ and $\sigma_0 \approx_j \sigma_1$,
	\item $(w,\sigma) \in V'(p)$ iff $w \in V(p)$.
\end{itemize}
Thus, the states of the updated model consist of those state-action pairs $(w,\sigma)$ such that the precondition of the action $\sigma$ is satisfied by the state $w$ in $M$; intuitively, $(w,\sigma)$ represents the state of the world $w$ after $\sigma$ has been performed. The definition of $V'$ ensures that such ``updated states'' satisfy the same primitive propositions as they did before (corresponding to the intuition that epistemic actions can change the \textit{information} agents have access to, but cannot change basic facts about the world). Finally, the definition of $\sim'_{j}$ specifies that updated state-action pairs are indistinguishable for agent $j$ precisely when the constituent states and actions were indistinguishable for $j$ in $M$ and $A$, respectively. It is easy to see that $M^A$ is an epistemic model.

The language for Action Logic, $\mathcal{L}_{AL}$, extends the basic epistemic language with an update operator:
$$\phi := p \, | \, \neg \phi \, | \, \phi \wedge \psi \, | \, K_j \phi \, | \, [A,\sigma] \phi.$$
$[A,\sigma] \phi$ is read ``if action $\sigma$ can be performed, then afterwards, $\phi$ is true''.\footnote{The ``if-then'' construction here is interpreted as a standard material conditional, i.e., $[A,\sigma]\phi$ is vacuously true when $\sigma$ cannot be performed.}
This language therefore lets us reason about agents' knowledge and how it can change as a result of epistemic actions. Formulas of $\mathcal{L}_{AL}$ can be interpreted in epistemic models as before, with the additional semantic clause for $[A,\sigma]$ given by:
$$
\begin{array}{lcl}
(M,w) \models [A,\sigma] \phi & \textrm{ iff } & (M,w) \models Pre(\sigma) \textrm{ implies } (M^A, (w,\sigma)) \models \phi.
\end{array}
$$
%Extending the language in this way does not forfeit preservation of truth via bisimilarity; given bisimilar epistemic models $(M,w)$ and $(M',w')$, they will satisfy the same formulas in $\mathcal{L}_{AL}$.
%Furthermore, given a particular epistemic state $(M,w)$, if two action models $(A,\sigma)$ and $(A',\sigma')$ are bisimilar, then the resulting models, $(M^A,(w,\sigma))$ and $(M'^{A'},(w',\sigma'))$, will be bisimilar.\footnote{Actually, bisimilarity of action models is not required for this result; this is also true for \emph{emulous} action models, a weaker requirement. This is explained in detail in (\cite{ditmarsch08}, section 6.5).}

To make these definitions clear, we present a simple example that will recur in various forms throughout the paper.
\begin{example} \label{exa:sim}
Colleagues Anne and Bob are discussing whether a particular company policy passed in this morning's board meeting: both are ignorant of whether or not the policy passed. Taking $p$ to be the proposition ``the policy passed'', this scenario may be represented with the simple epistemic model $M_{0}$ presented in Figure \ref{fgr:epi0}.\footnote{Unless stated otherwise, all relations are equivalence relations, so reflexive loops and edges implied by transitivity are assumed to be present, even when suppressed in the diagrams.}
\begin{figure}[h]
\centering
\includegraphics[scale=.35]{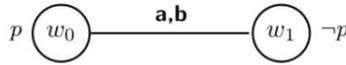}
\caption{$M_{0}$ -- Anne and Bob's initial uncertainty} \label{fgr:epi0}
\end{figure}

In this diagram and others like it, the circles represents the states of the model (in this case, $w_0$ and $w_1$), the formulas listed beside a state are true at that state (in this case, $p$ and $\neg p$ are true at $w_{0}$ and $w_{1}$, respectively), and edges between states are labelled with those agents that cannot distinguish those states (in this case, the edge labelled with $a$ and $b$ indicates that neither Anne nor Bob can distinguish between these two possible worlds).
Indeed, we have $M_{0} \models (\neg K_a p \wedge \neg K_a \neg p) \land (\neg K_b p \wedge \neg K_b \neg p)$, so this model properly captures the ignorance of Anne and Bob regarding $p$.

Along comes Carl, fresh from the board meeting with news of whether the policy passed.
He takes one of two actions: he either tells Bob that $p$ is true ($\sigma_{p}$), or he tells Bob that $p$ is false ($\sigma_{\lnot p}$). Anne watches intently: she knows that Carl is telling Bob what happened, but she is too far away to hear what Carl actually reports. We presume that Carl is honest, so the precondition of $\sigma_{p}$ is $p$, and the precondition of $\sigma_{\lnot p}$ is $\neg p$.
This is all captured in the action model $A_{0}$ depicted in Figure \ref{fgr:act0}, where the circles represent actions, the formulas listed beside the actions are the corresponding preconditions, and the edges represent action indistinguishability.
\begin{figure}[h]
\centering
\includegraphics[scale=.35]{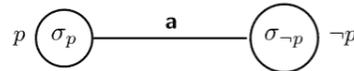}
\caption{$A_{0}$ -- Carl's communication to Bob} \label{fgr:act0}
\end{figure}

In particular, Anne cannot distinguish $\sigma_{p}$ from $\sigma_{\lnot p}$ (though she knows that \textit{one} of them happened), while Bob can. We therefore arrive at the following updated model $M_{0}^{A_{0}}$ given in Figure \ref{fgr:upd0}.
\begin{figure}[h]
\centering
\includegraphics[scale=.35]{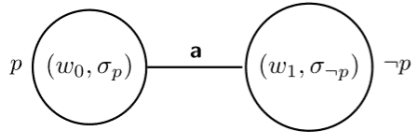}
\caption{$M_{0}^{A_{0}}$ -- After Carl's announcement} \label{fgr:upd0}
\end{figure}
Thus we see that when we update the epistemic model with the action model, the result is a model where Bob knows whether or not $p$, while Anne does not: $M_{0}^{A_{0}} \models K_b p \vee K_b \neg p$, and $M_{0}^{A_{0}} \models \neg (K_a p \vee K_a \neg p)$, as desired. \qed
\end{example}
%adam1
%Note the non-public nature of Carl's announcement;
Note that Carl's announcement is not \textit{public} since Anne cannot tell whether he announced $p$ or $\lnot p$ while Bob can.\footnote{This particular sort of non-public announcement is sometimes referred to as ``semi-private'' in the literature.}
This demonstrates that AL can capture epistemic dynamics that PAL cannot.%
%adam1: In light of reviewer 2's comment, I think the complaint is that we haven't shown that formulas of AL can't all be reduced to equivalent PAL formulas. I don't think this is hard to show, but I also don't think it follows trivially from example 1. So maybe our best bet here is to not make this claim.
%\footnote{In fact, AL is strictly more general than PAL:
\footnote{In fact, AL subsumes PAL:
to capture a public announcement of $\phi$, consider the action model $A_\phi$ consisting of a single node $\sigma$ with the precondition $\phi$. Then, given an epistemic model $M$ and a state $w$ therein, one can show that the resulting models $(M^{A_\phi},(w,\sigma))$ and $(M|_\phi,w)$ are bisimilar: updating with $A_\phi$ effectively deletes the states in $M$ where $\phi$ is false.
%adam1: removed for now
\commentout{
%adam1
On the other hand,
Carl's announcement
%adam1
in Example \ref{exa:sim}
%is an example of a dynamic which
cannot be captured by
%adam1
%PAL:
the dynamics of PAL, since
the update removes an edge from between $w_0$ and $w_1$, but does not delete any
%adam1
%nodes, which is something that PAL cannot do.}
nodes.
}
}

Nonetheless, as we have claimed, AL suffers from limitations of its own; we turn now to a discussion of these limitations.

%limitations
\subsection{Limitations of action logic} \label{sec:lim}

The limitations we have in mind
can be summarized quite simply: the AL framework treats the action model as common knowledge.
%turn on the update procedure that AL employs to produce
Indeed,
the updated model $M^A$
%In particular, the produced indistinguishability relations `hard-code'
imports much of the structure of the action model:
%the relations of the action model in such a way
%	that all agents come to know the structure of the action model.
in it, all agents come to know what actions the \emph{other} agents can distinguish. The result is that AL has difficulty capturing scenarios involving \emph{higher-order uncertainty}---e.g., uncertainty about what other agents know.
Although this limitation can be overcome, in a sense,
%through the design of action models,
by expanding the action model,
%but we will see that they will, in general be undesirably complex.
we will see that in general this is not an appealing solution.
Moreover, the examples we consider demonstrate that the AL framework is not as modular as it might appear to be: adjustments to the epistemic model will often require corresponding adjustments to the action model to preserve the intended semantic interpretations. To illustrate these points, we return to the scenario presented in Example \ref{exa:sim}.

\begin{example} \label{exa:fre0}
Suppose again that Carl comes along with news of $p$;
%after announcing that he has this news and plans to divulge it, however, he proceeds to deliver the information in French.
however, instead of speaking so that only Bob can hear him, Carl speaks plainly for all to hear, but he delivers the message in French. As it happens, Bob speaks French and Anne does not.

As before, the apparent actions that Carl might take are telling Bob that $p$ and telling Bob that $\neg p$, which presumably have preconditions $p$ and $\neg p$, respectively.
Bob can distinguish these actions, as he speaks French, while Anne cannot.
This reasoning produces the same action model $A_{0}$ depicted in Figure \ref{fgr:act0}, which therefore produces the same updated model $M_{0}^{A_{0}}$ shown in Figure \ref{fgr:upd0}.
%\begin{multicols}{2}
%\begin{center}
%\includegraphics[scale=.3]{diagrams/d6.eps}
%\captionof{figure}{}
%\end{center}
%\columnbreak
%\begin{center}
%\includegraphics[scale=.3]{diagrams/d7.eps}
%\captionof{figure}{}
%\end{center}
%\end{multicols}
As expected, then, just as in Example \ref{exa:sim}, Bob ends up knowing whether or not $p$, and Anne does not.
An unexpected result, however, is that Anne \textit{knows this former fact}, and Bob knows the latter!
That is, we have:
$$M_{0}^{A_{0}} \models K_a(K_b p \vee K_b \neg p)$$
and
$$M_{0}^{A_{0}} \models K_b(\neg (K_a p \vee K_a \neg p)).$$
The former says that Anne knows that Bob knows whether $p$, while the latter says that Bob knows that Anne is uncertain about $p$.
%Both of these are unanticipated consequences of the scenario we've laid out:
But there was no assumption that Bob knows that Anne cannot understand Carl's message, nor that Anne knows that Bob can. That is, we did not explicitly stipulate whether or not either knew about the other's (in)ability to speak French. To capture this, the model must be refined. \qed
\end{example}

%The update procedure works such that since
Loosely speaking,
Bob's ability to speak French and Anne's inability to speak French are represented in the structure of the action model $A_{0}$ in Figure \ref{fgr:act0}, and this is why they effectively become common knowledge in the updated model.
%adam: not sure this is worth including here
%\footnote{I use scare quotes here since common knowledge has a precise meaning that does not apply here: $\phi$ is common knowledge at a point just when for every sequence of $K$-operators, $\bar{K}$, $\bar{K}\phi$ is true at that point. Another way of stating the point I wish to make here is: all facts which follow from the structure of the action model become common knowledge in the updated model. This will be made sharper in the next subsection.}
%This doesn't seem desirable, as
But, of course,
we may want to capture a scenario where one or both are \textit{uncertain} about whether or not the other speaks French;
indeed, uncertainties of this sort play an important role in everyday reasoning.
%and an inability to model dynamics that involve them would be a formidable limitation.

For the sake of simplicity, let us begin by aiming only to
%modify our model so as to obviate only the second of the above undesired facts:
remove the consequence in the updated model
that Anne knows that Bob knows whether $p$.
\begin{example} \label{exa:fre1}
A relevant proposition here is that Bob speaks French;
call this $q$.
If we wish to account for Anne's uncertainty about $q$, we ought to expand the initial epistemic model $M_{0}$ to include the possible values this proposition might have. The result is the model $M_{1}$ depicted in Figure \ref{fgr:epi1}.
\begin{figure}[h]
\centering
\includegraphics[scale=.35]{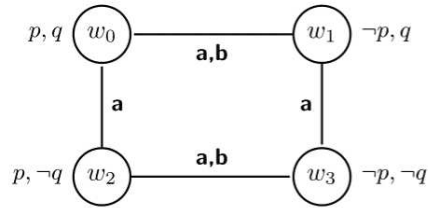}
\caption{$M_{1}$ -- An expanded model of Anne and Bob's initial uncertainty} \label{fgr:epi1}
\end{figure}

Anne does not know whether or not $p$ is true, and also does not know whether or not $q$ is true; thus, there are $a$-edges between all 4 nodes. 
We assume that Bob does know whether or not he speaks French, but, as before, does not know whether or not $p$; thus there are horizontal $b$-edges, but no vertical $b$-edges.

%adam*: do we want to add a picture for this?
It is easy to check that updating $M_{1}$ with $A_{0}$ produces an epistemic model $M_{1}^{A_{0}}$ in which Bob knows (at every state) the true value of $p$. So it seems we must modify the action model as well; we add
%Turning to the action model, we see that the inclusion of the proposition $q$ produces a third action, $\sigma$, deemed possible by Anne:
a third action, $\sigma$, corresponding intuitively to the ``unsuccessful'' announcement in which Carl speaks his piece but no one (including Bob) understands him.
The precondition for $\sigma$ should therefore be $\neg q$: that Bob does not speak French.
%(since our model effectively hard-codes the assumption that Anne cannot speak French)
Furthermore, the preconditions for the actions $\sigma_p$ and $\sigma_{\neg p}$ ought to be strengthened to include $q$, since these actions now represent Carl telling Bob $p$ or $\neg p$ \textit{and Bob understanding what was said}.
Bob will be able to distinguish any of these three actions, since he knows whether or not he speaks French, and, given that he speaks French, he knows which announcement Carl is making. Anne, on the hand, will not be able to distinguish any of the three actions---it all sounds the same to her.
All this is captured by the action model $A_{1}$ given in Figure \ref{fgr:act1}.
%adam*: here and elsewhere, I think the action model should label preconditions as single formulas, so, e.g., "p,q" should be replaced by "p \land q".
\begin{figure}[h]
\centering
\includegraphics[scale=.45]{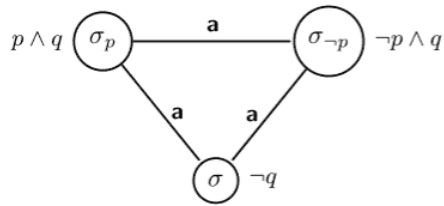}
\caption{$A_{1}$ -- An expanded action model} \label{fgr:act1}
\end{figure}

%Now, if we the update procedure on this reformulated static model and action model (figure 3.4), we produce the following static model (figure 3.5):

The updated model $M_{1}^{A_{1}}$ is shown in Figure \ref{fgr:upd1}.
\begin{figure}[h]
\centering
\includegraphics[scale=.4]{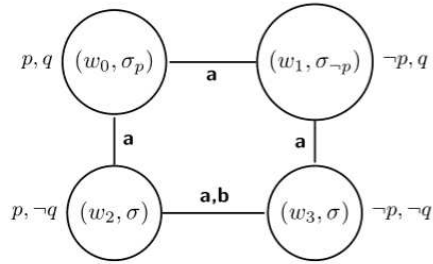}
\caption{$M_{1}^{A_{1}}$ -- After Carl's announcement} \label{fgr:upd1}
\end{figure}
%\pagebreak
%\begin{multicols}{2}
%\begin{center}
%\includegraphics[scale=.35]{diagrams/d9.eps}
%\captionof{figure}{}
%\end{center}
%\columnbreak
%\begin{center}
%\includegraphics[scale=.35]{diagrams/d10.eps}
%\captionof{figure}{}
%\end{center}
%\end{multicols}

As expected, we now have that Bob knows whether $p$ only in those nodes where $\sigma_p$ or $\sigma_{\neg p}$ was performed, and Anne does not know whether $p$ at all.
Furthermore, we have:
$$(w_0,\sigma_p) \models \hat{K}_a (K_b p \lor K_b \neg p) \land \hat{K}_a (\neg K_b p \land \neg K_b \neg p)$$
This reads: Anne considers it possible that Bob knows whether or not $p$, and also considers it possible that he does not.
%By building considering the relevant proposition $q$, we've recovered the uncertainty we were able to capture in our first pass.
Thus, by expanding the initial epistemic model to include Anne's uncertainty about Bob's ability to speak French, as well as adding a third node to the action model corresponding (roughly speaking) to an ``unsuccessful" announcement from Carl, we are able to capture the second-order uncertainty that we set out to capture. \qed
\end{example}

One might reasonably feel some discomfort regarding the introduction of $\sigma$ into the action model. On at least one intuition for what constitutes an ``action'', Carl's announcement (in French) of $p$ ought to count as the \textit{same} action regardless of who hears it or what languages they might understand. In other words, one might object to distinguishing $\sigma_{p}$ from $\sigma$ on the grounds that it builds into the ontology of actions properties that really have nothing to do with actions, but rather with agents.

This philosophical objection could perhaps be swept aside if the underlying formalism actually did the job we wanted it to: it is a hard case to make to let vague ontological concerns trump mathematical efficacy. What we now aim to demonstrate, however, is that this technique of expanding the action model is \textit{not} an effective tool for capturing higher-order uncertainty.

\commentout{
Let us alter the scenario of Examples \ref{exa:fre0} and \ref{exa:fre1} in only one respect: we assume now that Anne already knows whether or not $p$ is true before Carl says anything.
As before, however, Bob is ignorant of $p$, and neither Anne nor Bob knows whether the other can understand French.
The corresponding epistemic model $M_{2}$ is given in Figure \ref{fgr:epi2}.
\begin{figure}[h]
\centering
\includegraphics[scale=.35]{diagrams/d11.eps}
\caption{Anne starts off knowing whether $p$} \label{fgr:epi2}
\end{figure}

How should we adjust the action model $A_{1}$?
Since Anne already knows whether or not $p$ is true (and assumes Carl is truthful), it would seem that Anne can distinguish $\sigma_{p}$ from $\sigma_{\neg p}$.
It is unclear, however, whether she can distinguish either of these actions from $\sigma$.
$\sigma$ represents the action where Carl's announcement goes unappreciated by Bob, because he does not know French.
Since Anne cannot tell whether Bob understood Carl's announcement or not, it would seem that we should keep the $a$-edge between $\sigma_{p}$ and $\sigma$; but since Anne knows whether or not $p$ holds, she \textit{can} distinguish between an announcement of $p$ understood by Bob, and an announcement of $\neg p$ which is not by understood by Bob -- this consideration stands in favor of removing the $a$-edge. 
The problem we see is that the contrived nature of $\sigma$ -- an action which encodes its effects on a particular agent -- renders impossible the representation of any scenario featuring the indistinguishability relations of any other agents. 

The fix to this problem is another adjustment to the action model.
We split the action $\sigma$ into two actions $\sigma_{\neg p \neg q}$ and $\sigma_{p \neg q}$, and relabel the other actions appropriately.
These new nodes represent the actions \textit{Carl announces that $\neg p$ and Bob does not understand} and \textit{Carl announces that $p$ and Bob does understand}.
The adjusted action model is given in Figure \ref{fgr:act2}, and the corresponding updated model is given in \ref{fgr:upd2}.

\begin{figure}[h]
\centering
\includegraphics[scale=.35]{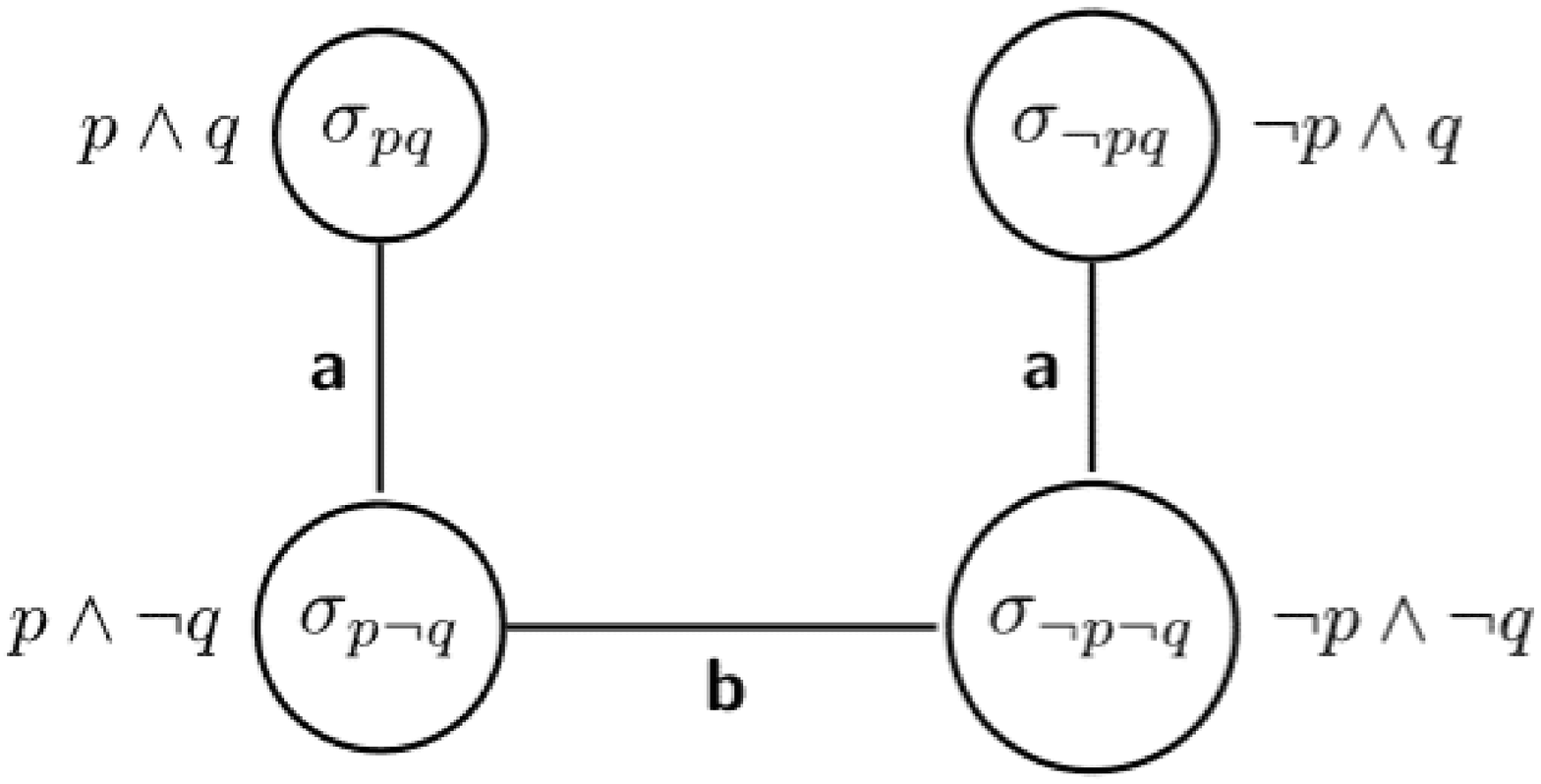}
\caption{The adjusted action model for when Anne knows whether $p$.} \label{fgr:act2}
\end{figure}

The resulting updated model will then be:

\begin{figure}
\centering
\includegraphics[scale=.35]{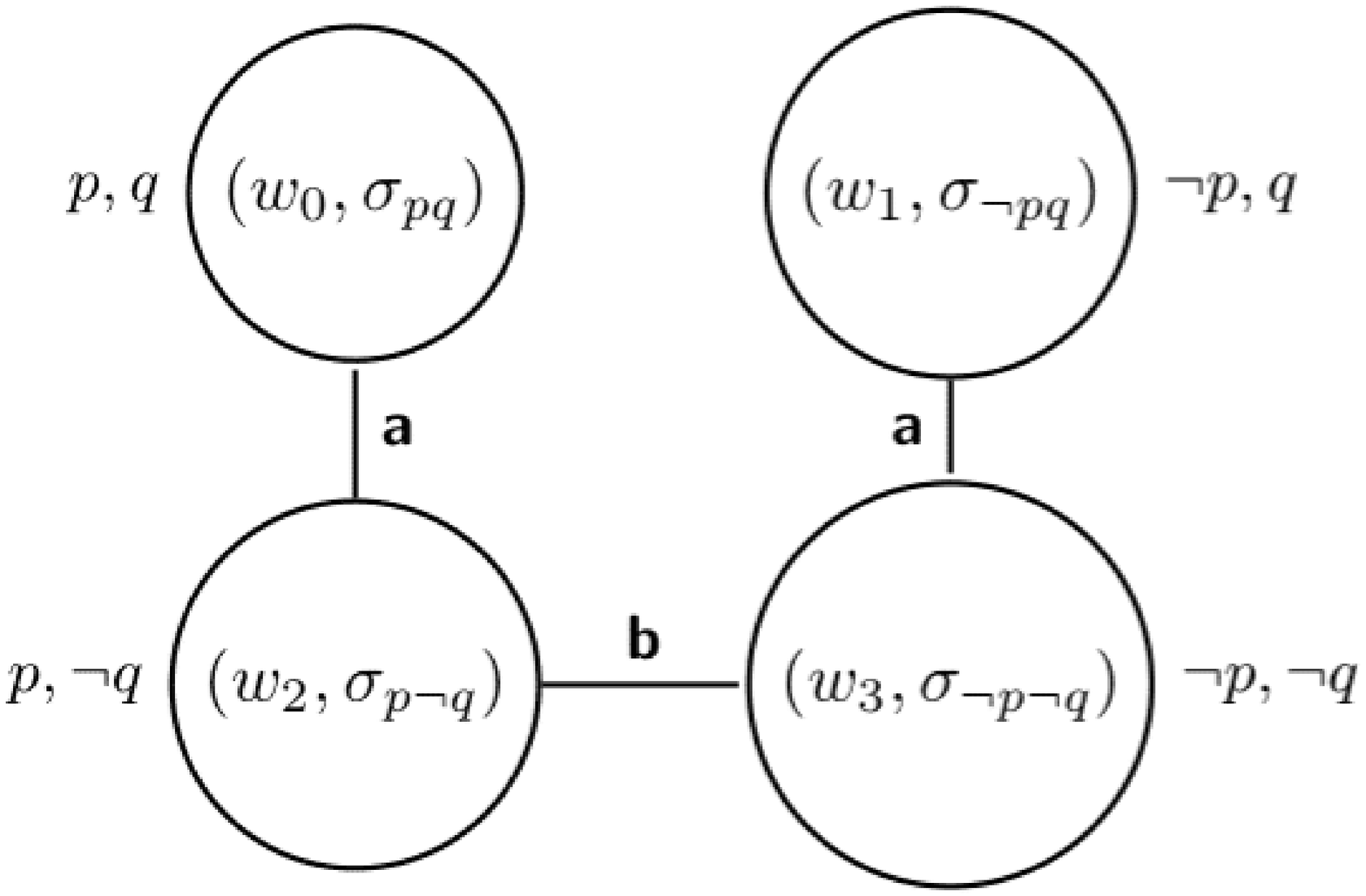}
\caption{The updated model when Anne knows whether $p$.} \label{fgr:upd2}
\end{figure}

As expected, Anne knows whether or not $p$ in the updated model, but does not know whether or not Bob speaks French. 
With suitable adjustments to the action model, we were able to represent the scenario, but consider what change necessitated these adjustments: we only changed Anne's epistemic state with respect to a proposition, $p$.
That is, a minor change to the prior knowledge of an agent required that we recast the actions represented. 
Not only does this demonstrate an oddity of the framework, but it also demonstrates a major inefficacy: following all but the most trivial of adjustments to the epistemic scenario, the action model will require revision.
In this example in particular, the construction of the action model essentially consisted in the computation of the updated model that it would produce.
}

\begin{example} \label{exa:fre2}
We alter the scenario of Examples \ref{exa:fre0} and \ref{exa:fre1} in only one respect: we assume now that Anne \textit{does} speak French. This requires no change to the intial epistemic model $M_{1}$ (since whether or not Anne speaks French is not represented explicitly in this model), but it does, intuitively, require us to re-work the action model $A_{1}$. In particular, the Anne-edge connecting $\sigma_{p}$ and $\sigma_{\lnot p}$ no longer seems appropriate, since now Anne \textit{can} understand what Carl announces.

Now we are faced with a somewhat awkward question---should there be an Anne-edge between $\sigma_{p}$ and $\sigma$? Intuitively, there should be, since $\sigma$ is supposed to encode the fact that Bob does not understand Carl's announcement, and Anne is not supposed to be able to tell whether he does or not. Similar reasoning leads us to leave the Anne-edge between $\sigma_{\lnot p}$ and $\sigma$ in place, so the resulting relation fails to be transitive.

Perhaps we could relax the requirements placed on the relations $\approx_{i}$ in action models to accommodate this type of problem, but in fact there is a deeper issue here that suggests an alternative resolution: it is easy to see that \textit{any} reflexive relation $\approx_{a}'$ for Anne on the set $\{\sigma_{p}, \sigma_{\lnot p}, \sigma\}$ produces an action model $A_{1}'$ such that the updated model $M_{1}^{A_{1}'}$ satisfies $(w_{2},\sigma) \sim_{a} (w_{3},\sigma)$. But this misrepresents the situation: in world $w_{2}$, Carl's announcement must have been that the policy passed, $p$; as such, after his announcement, Anne should no longer be uncertain about $p$.

The problem here lies with $\sigma$: it was introduced originally to represent the possibility of an ``unsuccessful" announcement by Carl. But in the present context, Carl's announcement is always at least partially successful, in that it always informs Anne of the truth value of $p$. The natural fix to this problem is another adjustment to the action model: we ``split'' the action $\sigma$ into two actions, $\sigma_{\neg p \neg q}$ and $\sigma_{p \neg q}$ (and relabel the other actions for clarity). The new action model $A_{2}$ is given in Figure \ref{fgr:act2}.
\begin{figure}[h]
\centering
\includegraphics[scale=.35]{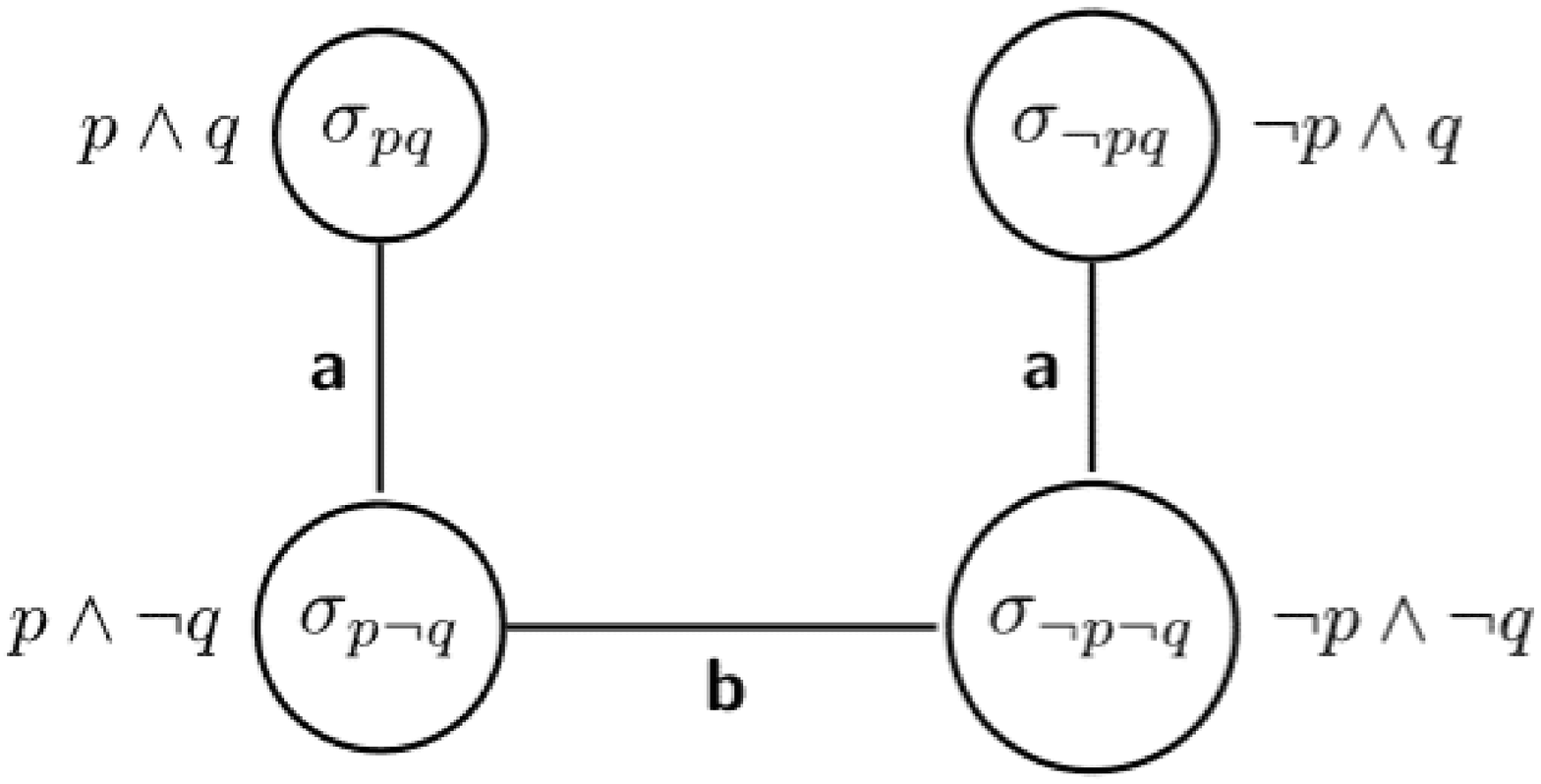}
\caption{$A_{2}$ -- The adjusted action model for when Anne speaks French.} \label{fgr:act2}
\end{figure}

These new actions $\sigma_{p\lnot q}$ and $\sigma_{\lnot p \lnot q}$ might be thought of as corresponding to situations where Carl announces $p$ and Bob does not understand, and where Carl announces $\lnot p$ and Bob does not understand, respectively. Anne can distinguish announcements based on their content, but not based on whether Bob understands them. Bob can distinguish announcements based on whether he understands them and, provided he understands them, based on their content as well. The updated model $M_{1}^{A_{2}}$ is given in Figure \ref{fgr:upd2}.
\begin{figure}[h]
\centering
\includegraphics[scale=.35]{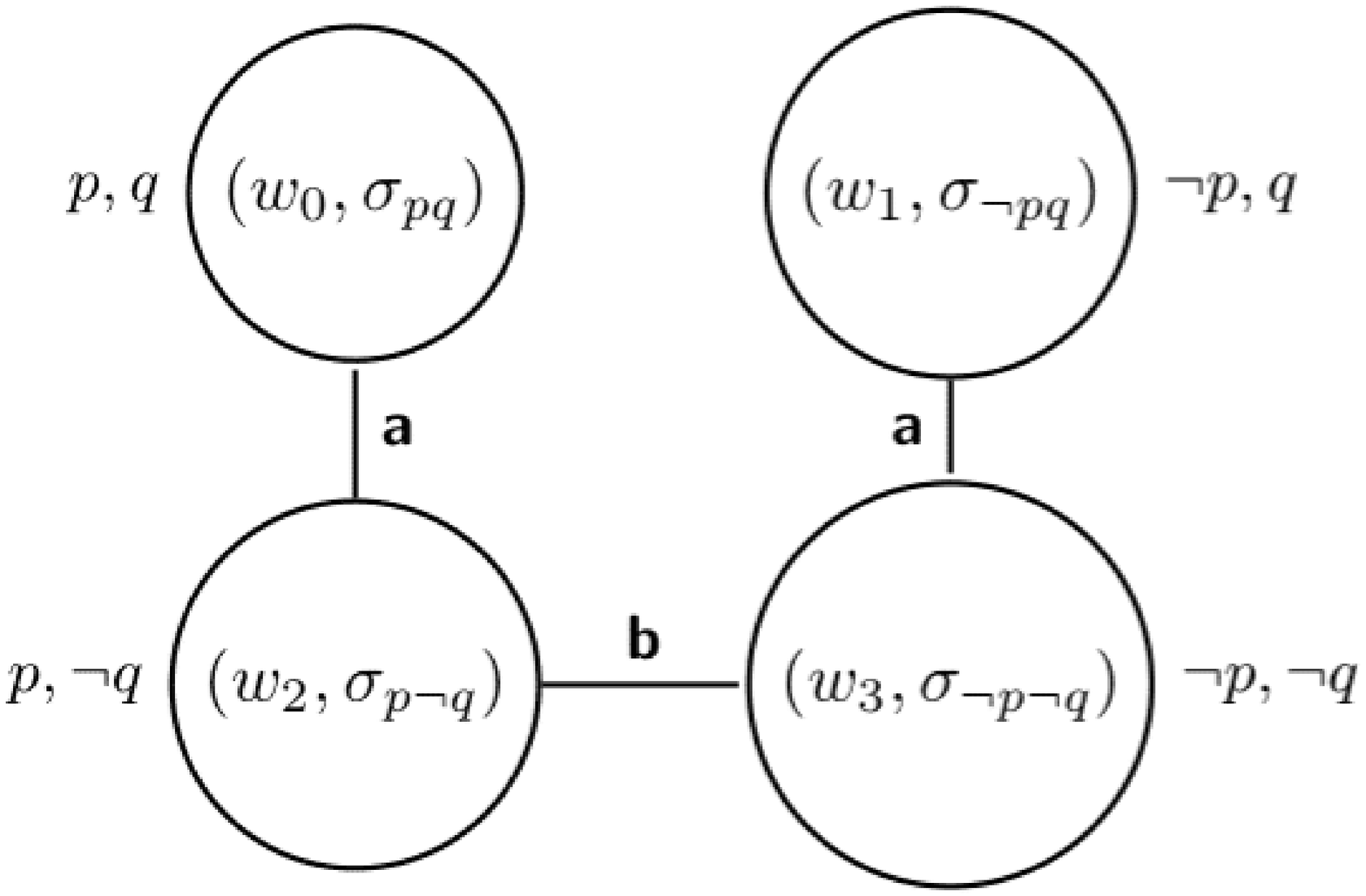}
\caption{$M_{1}^{A_{2}}$ -- The updated model when Anne speaks French.} \label{fgr:upd2}
\end{figure}
As expected, Anne has learned whether or not $p$ in the updated model (since she heard Carl), but she continues to be ignorant as to whether or not Bob speaks French (and, in turn, whether or not Bob has learned $p$). \qed
\end{example}

Note the duplication of effort in the construction of $A_{2}$. The initial epistemic model $M_{1}$ already encodes the possibilities regarding Bob's ability to speak French and Anne's uncertainty about this. Yet our action model recapitulates this structure with actions that incorporate not just what Carl says, but also whether Bob understands it or not. Moreover, once our background assumptions are fixed (such as whether Anne understands French or not), edges (and nodes!) in $A_{2}$ are determined, essentially, by examining $M_{1}$ and reading off what the uncertainties ought to be. Thus, while AL gives the impression of a clean, modular division between epistemic states and actions, in practice the two seems to be quite tangled, with unavoidable redundancies in their representations. To drive this point home, we sketch one further example.

\begin{example} \label{exa:fre3}

Consider an expanded epistemic model $M_{2}$ in which we take not just Bob's but also \textit{Anne's} knowledge of French as endogenous: that is, suppose we also wish to represent Bob as being uncertain of whether or not Anne speaks French.
A simple model of such a scenario might consist in eight worlds representing the possible combinations of truth values for primitive propostions $p$, $q$, and $r$, where $p$ and $q$ are interpreted as before and $r$ stands for the proposition ``Anne speaks French''.

The action model $A_{2}$ of Example $\ref{exa:fre2}$ is again inadequate. To see why, consider whether there ought to be an Anne-edge connecting $\sigma_{p}$ and $\sigma_{\lnot p}$. Intuitively, whether Anne can distinguish Carl announcing $p$ from Carl announcing $\lnot p$ depends on the world (i.e., it depends on whether or not Anne can speak French); it is not a fixed and unchanging truth that can be hard-coded into the model. As such, in order to capture this with a fixed action model we require, yet again, a proliferation of actions: e.g., actions of the form $\sigma_{pq\neg r}$, corresponding to something like Carl announcing $p$ and Bob but not Anne understanding it. \qed
%, however, would admit the representation of some odd actions.
%For reasons similar to those stated above, an action model consisting in $\sigma_{p}$ and $\sigma_{\neg p}$ would not suffice.
%This is because whether either agent can distinguish the actions would depend on whether they speak French.
%Thus, actions would be of the form $\sigma_{pq\neg r}$: \textit{Carl announces that $p$ and Bob understands it, but Anne does not}.
\end{example}

These examples make it clear that the AL formalism is not well-suited to the practical task of building models to represent scenarios in which second-order knowledge is relevant: in addition to specifying the initial epistemic model, one must construct alongside it an elaborate space of actions fine-tuned to the specifics of the epistemic setting. These actions, rather than corresponding in a natural way with concrete events in the world (like Carl making an announcement), are individuated by details about agents' perceptions of them that seem less like part of the actions themselves and more like part of the background epistemics of the situation.
%As this final example makes explicit, it is Action Logic's failure to represent that an agent's indistinguishability of actions may be world-dependent which renders it inefficient.
We turn now to our proposed solution to this problem, which simplifies the action space considerably and imports the representation of higher-order uncertainty into the epistemic model, which is designed to handle it.

\commentout{
\subsection{Recovery}

This example and theorem point to a problem with the AL framework. 
In particular, the problem seems to be with the update procedure;
	the procedure forces the updated model to be such that all agents are made aware of the structure of the action model.
While this may be avoided by encoding uncertainty into the action model, this doesn't seem faithful to the spirit of AL.
One way out might be an adjustment of the AL framework and update procedure:
	instead of taking the relations between actions to be independent of the worlds in which they are performed;
	one might take these relations to be dependent on what is true in these worlds.
This might be realized by assigning each pair of actions $(\sigma_0,\sigma_1)$ a pair of propositions $(\phi_0,\phi_1)$ which
	need to be true in a pair of worlds in order for the updated worlds resulting from the execution of the actions in the respective worlds
	to be indistinguishable.
This will be pursued in future work.
}
%
%
%
%

%adjusted semantics
\section{Adjusted Semantics} \label{sec:adj}

%We propose an adjusted semantics to DEL which avoids the explosion of states in the action model. In fact, because our semantics essentially make the distinguishability of the actions world-dependent, it is natural in our revised framework to drop the separation between static and action models.
We propose a new semantics for modeling information update that subsumes AL and is able to capture higher-order uncertainty without the proliferation of actions illustrated in the previous section. In essence, we endogenize the action model, making the distinguishability of the actions world-dependent. It is therefore natural in our revised framework to drop the notion of a separate action model altogether.

In addition to a countable collection of primitive propositions $\textsc{prop}$ and a finite set of agents $G$, fix a set $\Sigma$ of \emph{epistemic actions}
%adam2: added; we need this to not be part of the model. I also took Pre out of the model definition, below
together with a \emph{precondition} function $Pre: \Sigma \to \mathcal{L}_{EL}$ specifying the precondition for each action as before.
%Let $\Sigma$ be a fixed set; intuitively, the set of all actions.
A \textbf{dynamic model (over $(\Sigma, Pre)$)} is a tuple
$$M = \langle W, \{\sim_j : j\in G\}, \{f_j: j \in G\}, V \rangle$$
where $\langle W, \{\sim_j : j\in G\}, V \rangle$ is an epistemic model, and for each agent $j$, $f_j : W \rightarrow 2^{\Sigma \times \Sigma}$ is a function from worlds to relations on actions. Intuitively, $(\sigma_1,\sigma_2) \in f_j(w)$ means that at world $w$, if action $\sigma_1$ is performed, then agent $j$ cannot rule out $\sigma_2$ being the action performed. The $f_j$ functions are the crucial novel component of our framework which allow action indistinguishability to vary from world to world.

We will use dynamic models to interpret the language $\mathcal{L}_{DL}$ recursively defined by
$$\phi := p \, | \, \neg \phi \, | \, \phi \wedge \psi \, | \, K_j \phi \, | \, [\sigma] \phi,$$
where $[\sigma] \phi$ is read, as before, ``if action $\sigma$ can be performed, then afterwards, $\phi$ is true''.%
\footnote{We no longer need to ``tag'' the update modality with an action model since our framework does not employ action models; this turns out to have important implications in our proof of completeness.}
Of course, in order to interpret the update modalities, we must define the notion of an updated model.
Let
%adam2: removed Pre^{+}
$$M^{+} = \langle W^{+}, \{ \sim^{+}_j: j \in G\}, \{ f^{+}_j: j\in G \}, V^{+} \rangle,%
%adam2: added
\footnote{This notation does not specify \textit{which} action $\sigma$ the update is being performed with respect to, since (as in AL) our update procedure effectively performs all available updates simultaneously.}$$
where:
\begin{eqnarray*}
W^{+} & = & \{(w,\sigma) \: : \: w \models Pre(\sigma) \}\\
(w,\sigma) \sim^{+}_j (w',\sigma') & \dimp & w\sim_j w' \text{ and } (\sigma,\sigma')\in f_j(w)\\
f^{+}_j(w,\sigma) & = & f_j(w)\\
%adam2
%Pre^+(\sigma) & = & Pre(\sigma)\\
(w,\sigma)\in V^{+}(p) & \dimp & w\in V(p).
\end{eqnarray*}
It is easy to see that $M^{+}$ is itself a dynamic model over $(\Sigma, Pre)$ provided the relations $\sim_{j}^{+}$ are equivalence relations. In general, this need not be the case; however, the following conditions guarantee that it is.
%$M^{+}$ is again a dynamic model---in particular, the relations $\sim_{j}^{+}$ are equivalence relations---when the following conditions are satisfied:\footnote{We use the notation $X_{min}$ to denote the minimal set $X'$ such that $X\subseteq X'$ and $X'$ is an equivalence relation.}
\begin{description}
	\item[(C1)] If $w\sim_j w'$, then $f_j(w) = f_j(w')$.
	\item[(C2)] Each $f_j(w)$ is an equivalence relation on $\Sigma$.
\end{description}

\begin{proposition} \label{pro:dyn}
(C1) and (C2) together imply that $M^{+}$ is a dynamic model.
\end{proposition}

\begin{proof}
As noted, it suffices to show that $\sim^{+}_j$ is an equivalence relation. Reflexivity follows immediately from reflexivity of $\sim_{j}$ and $f_{j}(w)$. For symmetry, suppose that $(w,\sigma)\sim^+_j (w',\sigma')$. Then $w \sim_j w'$ so also $w'\sim_j w$; moreover, by (C1) we have $(\sigma,\sigma') \in f_{j}(w) = f_{j}(w')$, so (C2) implies $(\sigma',\sigma) \in f_{j}(w')$, whence $(w',\sigma') \sim^+_j (w,\sigma)$. For transitivity, suppose that $(w,\sigma) \sim^+_j (w',\sigma')$ and $(w',\sigma') \sim^+_j (w'',\sigma'')$. Clearly $w \sim_{j} w''$. Moreover, we have $(\sigma,\sigma') \in f_{j}(w)$ and $(\sigma',\sigma'') \in f_{j}(w')$; by (C1) $f_{j}(w) = f_{j}(w')$, and by (C2) this relation is transitive, so we deduce that $(\sigma,\sigma'') \in f_{j}(w)$, which shows that $(w,\sigma) \sim^+_j (w'',\sigma'')$, as desired.%
%adam: trying this out as a footnote. what do you think?
\footnote{In fact, condition (C1) is stronger than it needs to be to establish Proposition \ref{pro:dyn}: since $\sim_j$ is assumed to be symmetric; weakening (C1) to:
\begin{description}
\item[(C1$^*$)]
If $w \sim_j w'$, then $f_j(w) \subseteq f_j(w')$
\end{description}
has no effect. Of course, this argument would not work for weaker epistemic logics where agents are not assumed to be negatively introspective and $\sim_{j}$ is not assumed to be symmetric. Exploring the interplay between properties of the relations $\sim_{j}$ and properties of the functions $f_{j}$ is intriguing, but takes us too far afield in the present work.}
%\begin{itemize}
%	\item For any $(w,\sigma)$, we know that $w\sim_j w$. Since $f_j(w)$ is reflexive-closed, $(\sigma,\sigma)\in f_j(\Phi_w)$, and so $(w,\sigma) \sim^f_j (w,\sigma)$.
%	\item Suppose that $(w,\sigma)\sim^f_j (w',\sigma')$. Then $w\sim_jw'$ and so since $\sim_j$ is an equivalence relation, $w'\sim_j w$. Also, by (C1), $f_j(w) = f_j(w')$. Since $(\sigma,\sigma')\in f_j(w)$, this implies that $(\sigma,\sigma')\in f_j(w')$. Then since $f_j(w')$ is symmetric-closed, $(\sigma',\sigma) \in f_j(w')$. We conclude that $(w',\sigma') \sim^f_j (w,\sigma)$.
%	\item Suppose that $(w,\sigma)\sim^f_j (w',\sigma')$ and $(w',\sigma')\sim^f_j (w'',\sigma'')$. Then by the transitivity of $\sim_j$, we have that $w\sim_j w''$, so it remains to show that $(\sigma,\sigma'')\in f_j(w)$. Since $f_j(w)$ is closed under transitivity, it suffices to show that $(\sigma',\sigma'')\in f_j(w)$. Since $w\sim_j w'$, (C1) implies that $f_j(w) = f_j(w')$, and so since $(\sigma',\sigma'')\in f_j(w')$, we have that $(\sigma,\sigma'')\in f_j(w)$, which concludes our proof.
%\end{itemize}
\end{proof}

Henceforth, we assume (C1) and (C2). Now define
$$
\begin{array}{lcl}
(M,w) \models [\sigma] \phi & \textrm{ iff } & (M,w) \models Pre(\sigma) \textrm{ implies } (M^+, (w,\sigma)) \models \phi.
\end{array}
$$
This completes our specification of the new semantics for epistemic actions. It is not hard to see that update by an action model is essentially a special case of update in this framework: given an action model $A$, one simply defines each $f_{j}$ to be the constant function such that $f_{j}(w) = {\approx_{j}}$.\footnote{Successive updates by different action models $A_{1}, \ldots, A_{n}$ can be recast as repeated updates by the same action model $A$, where $A$ is an appropriate disjoint union of the $A_{k}$'s.} This directly realizes the intuition that the action model is common knowledge in AL.

%Before investigating its formal properties and its relationship to AL, we return to our motivating examples to show that this new framework actually addresses the deficiencies we demonstrated.

%We now impose one more condition on $f$. $f_j$ is \textit{finitely specified} if there is a finite set of mutually exclusive formulas, $\Phi_j$ such that every world %satisfies exactly one member of $\Phi_j$. 
%\[ (C3) \;\;\; \text{For every } j \text{ in } G \text{, } f_j \text{ is finitely specifiable.} \]
%It is an immediate consequence that for each $j$, $f_j$ can be understood as operating on formulas instead of theories. 
%More precisely, for each $j$, we have a formula $\phi_{(j,\sigma,\sigma')}$ such that:
%\[ w \vDash \phi_{(j,\sigma,\sigma')} \text{ if and only if } (\sigma,\sigma') \in f_j(\Phi_w) \]
%We take this restriction to be an intuitive one suggested by each of our examples considered so far.
%For instance, in the example represented in figure 3.4, $f$ depends only on $q$ -- $f_a$ is constant, and which actions Bob can distinguish depends only on the truth value of $q$. 
%In these examples, the disinguishability of actions depends on only finitely many facts at any world.

%adam: this should go later
%
%
%
%
\commentout{
Consider now the following strengthened version of (C1):
\begin{description}
\item[(C)]
$(\forall j \in G)(\forall w,w' \in W)(f_j(w) = f_j(w'))$.
\end{description}
Intuitively, this gives us back the original action logic framework: (C) insists that the relations between actions be invariant across worlds, so $f_j$ can be identified with a single relation on $\Sigma$. This allows us to use the $f_{j}$ functions to construct an action model that produces a semantically equivalent notion of update. To make this precise, we must first define the semantics for updates in our new setting.
%adam*: this should be a proposition

We expand the language of AL to $\mathcal{L}_{DL}$ with an additional formula $\phi_{j,\sigma,\sigma'}$ for every $j, \sigma, \sigma'$; we denote by $\mathcal{L}_{EL^+}$ the fragment of $\mathcal{L}_{DL}$ which does not include the action model update operator.
The semantics for the new formulas are as follows:

\[ M,w \vDash \phi_{j,\sigma,\sigma'} \text{ iff } (\sigma,\sigma') \in f_j(w) \]
}
\commentout{
%translation
\subsection{Translation}

In this section, we develop a \emph{translation method} which takes one of our dynamic models, $M$, and returns an action model $A$ such that the update model produced by $A$ and $M_s$, the static fragment of $M$, is bisimilar to $M^+$.

Consider any $M$; we build $A = \langle \Sigma', \{\approx_j|j\in G\}, Pre' \rangle$.
For every world $(w_x,\sigma_y)\in M^+$, we create an action $\sigma_{xy} \in \Sigma$. We expand $\mathcal{L}$ to $\mathcal{L}'$ by introducing a new proposition $p_{x}$ for every world $w_x$ such that $p_x$ is only true at that world.\footnote{Note that specifying the realm of truth in $W$ for new propositions does not change $M$ in any way that doesn't respect bisimilarity. For this reason, we denote $M$ with these new propositions in the same way.}
We define $Pre'$ by $Pre(\sigma_{xy}) = p_{x}$; this implies that $\sigma_{xy}$ may only be performed at $w_x$.
Lastly, we define $\approx_j$ as follows:

\[ \sigma_{xy} \approx_j \sigma_{x'y'} \text{ iff } (\sigma_y,\sigma_{y'})\in f_j(w_x) \]

We now verify that $M^A_s$ and $M^+$ are bisimilar with respect to the language $\mathcal{L}$.
We define a relation $B$ between the worlds of the models as follows:\\[-10mm]

\[ (w_x,\sigma_{xy})B(w_x,\sigma_{y}) \]

The two nodes trivially satisfy the same propositions in $\mathcal{L}$.
For the forth condition, suppose that $(w_x,\sigma_{xy}) \sim'_j (w_{x'},\sigma_{x'y'})$.
Then: \\[-10mm]

\[ w_x \sim_j w_{x'} \] \\[-10mm]

and \\[-10mm]

\[ \sigma_{xy} \approx_j \sigma_{x'y'} \]  \\[-10mm]

So, by the way that $\approx_j$ is defined:  \\[-10mm]

\[(\sigma_y,\sigma_{y'})\in f_j(\Phi_{w_x}) \]  \\[-10mm]

From these facts we conclude that: \\[-10mm]

\[(w_x,\sigma_y)\sim^f_j (w_{x'},\sigma_{y'})\] \\[-10mm]

in $M^+$. 
For the back condition, suppose that:  \\[-10mm]

\[ (w_x,\sigma_y)\sim^f_j (w_{x'},\sigma_{y'}) \]

This implies that:

\[w_x \sim_j w_{x'}\]

and 

\[(\sigma_y,\sigma_{y'})\in f_j(w_x) \]

By the definition of $\approx$, we have that:

\[ \sigma_{xy} \approx_j \sigma_{x'y'} \]

and by the update procedure, we have that:

\[ (w_x,\sigma_{xy})\sim'_j (w_{x'},\sigma_{x'y'}) \]

This completes the proof that $B$ is a bisimulation.
Note that each of the nodes in $A$ will only be executable in one world in $M_s$.
This is by design; since $\Sigma'$ is, in effect, a copy of $W^+$, we desire that the result of the new update procedure replicate $A$ as a static model.
	
This translation demonstrates that our framework does not allow for the construction of update models which could not be the result of the old update method: for any dynamic model, $M$, there is some action model $A$ such that $M^A_s$ and $M^+$ are bisimilar.
The $A$ produced by the translation method here is clearly `overkill': to produce the correct update model, we require an action model $A$ that is already the size of the update model we wish to produce -- $A$ provides no insight into the action being performed, but rather describes the result of the action.
Note that, in some cases, we need not replicate the model entirely: if any two nodes $(w_x,\sigma_y)$ and $(w_{x'},\sigma_{y'})$ in $M^+$ share a j-edge for every $j\in G$, then we collapse $\sigma_{xy}$ and $\sigma_{x'y'}$ into the same action. 
This is because, for every agent, the two resulting nodes are indistinguishable, so those two actions are -- as performed in those worlds -- indistinguishable.

Now, one might ask, what is so bad about this proliferation of states?
Why should we care that $A$ ends up being a replication of $M^+$.
We should care because this demonstrates that action models offer no simplification of the situation resulting from the dynamics.
$M^A$ should purportedly be the thing explained by our construction of $A$ -- it should offer some insight into the structure of the dynamics.
What we find, however, is that in building $A$, we end up building $M^A$ itself -- the explanation we offer for $M^A$ turns out just to be $M^A$ itself. 
}
%
%
%
%

%revisiting
\subsection{Revisiting the examples}

We now revisit the problematic examples of Section \ref{sec:lim} and show that the new framework we have developed actually addresses the deficiencies we demonstrated.
Let $\Sigma = \{\sigma_{p}, \sigma_{\lnot p}\}$ and set $Pre(\sigma_{p}) = p$ and $Pre(\sigma_{\lnot p}) = \lnot p$, corresponding to the two intuitive actions of Carl announcing $p$ or announcing $\lnot p$, respectively.
Recall that in Examples \ref{exa:fre0} and \ref{exa:fre1}, Carl delivers his message in French, we assume that Anne does not know French, and we let $q$ represent the proposition that Bob knows French. We therefore define the dynamic model $\tilde{M}_{1}$ that we will use to reason about this scenario by extending the epistemic model $M_{1}$ depicted in Figure \ref{fgr:epi1}. In particular, we set
%In the dynamic model framework, $M$ is the following tuple:
%\[ \langle W, \{\sim_j|\; j\in G\}, \{f_j |\; j \in G\}, Pre, V \rangle \]
%where, letting $\sigma_p$ be the action `Carl announces that $p$', and $\sigma_{\neg p}$ be the action `Carl announces that not-$p$', and letting $q$ be the proposition that Bob knows French, we have the following:
%\[ W = \{w_i|\; 0\leq i \leq 3\} \]
%\[ \{\sim_j|\; j\in G\} = \{(w_0\sim_a w_1),(w_1\sim_a w_2),(w_2\sim_a w_3),(w_0 \sim_b w_1), (w_2 \sim_b w_3) \}_{min} \]
%% I'm abusing notation here to denote that f_a(w_0)=...=f_a(w_3)=\Sigma \times \Sigma
%\[ f_a(W) = \Sigma \times \Sigma \]
%\[ f_b(w_0)=f_b(w_1) = \emptyset_{min} \]
%\[ f_b(w_2)=f_b(w_3) = \Sigma \times \Sigma \]
%\[ Pre(\sigma_p) = p \;\;\;\; Pre(\sigma_{\neg p}) = \neg p \]
%\[ V(p) = \{w_0,w_2 \} \;\;\;\; V(q) =  \{w_0,w_1 \} \]
\begin{itemize}
\item
$f_a(w_0) = f_{a}(w_{1}) = f_{a}(w_{2}) = f_a(w_3) = \Sigma \times \Sigma$,
\item
$f_{b}(w_{0}) = f_{b}(w_{1}) = id_{\Sigma}$, and
\item
$f_{b}(w_{2}) = f_{b}(w_{3}) = \Sigma \times \Sigma$,
\end{itemize}
where $id_{X}$ denotes the identity relation on $X$. Thus, this dynamic model encodes the fact that Anne can never distinguish the two actions, whereas Bob can distinguish them just in case he speaks French.
%The static part of the model remains the same; what changed is that we need only describe the effects of two actions, `announce $p$', and `announce not-$p$'.
%We may have coarser actions because we've imported some of the complexity into our indistinguishability functions: now our actions may be distinguished differently in different worlds.
%Together, the dynamic and static parts of $M$ provide enough information to produce the static model, $M^+$, which can be represented as the following diagram:
It is easy to see that (the epistemic part of) $\tilde{M}_{1}^{+}$ looks exactly like the model $M_{1}^{A_{1}}$ depicted in Figure \ref{fgr:upd1}, except with the nodes $(w_{2},\sigma)$ and $(w_{3},\sigma)$ relabeled $(w_{2}, \sigma_{p})$ and $(w_{3},\sigma_{\lnot p})$, respectively. In other words, our update produces the ``right'' epistemic results, and it does so using a simple and natural set of actions and without requiring any fine-tuning of the model beyond the basic association between worlds where Bob speaks French and worlds where he can distinguish Carl's two possible announcements.

%\begin{center}
%\includegraphics[scale=.35]{diagrams/d19.eps}
%\captionof{figure}{}
%\end{center}
%$M^+$ is bisimilar to the updated model produced above, showing that our dynamic model captures the same effect of the action.

Next consider the scenario of Example \ref{exa:fre2}, which is just like the previous one except it is assumed that Anne \textit{does} speak French. To capture this, we need only change one line of the previous specifications for $\tilde{M}_{1}$:
\begin{itemize}
\item
$f_a(w_0) = f_{a}(w_{1}) = f_{a}(w_{2}) = f_a(w_3) = id_{\Sigma}.$
\end{itemize}
This corresponds directly to the assumption that Anne knows French (i.e., this is valid in the model), so she can always distinguish the two actions in question. Call this dynamic model $\tilde{M}_{1}'$. Now as before, it is straightforward to check that (the epistemic part of) $\tilde{M}_{1}'^{+}$ looks exactly like the model $M_{1}^{A_{2}}$ given in Figure \ref{fgr:upd2}, provided we replace every instance of $\sigma_{pq}$ and $\sigma_{p \lnot q}$ with $\sigma_{p}$, and every instance of $\sigma_{\lnot p q}$ and $\sigma_{\lnot p \lnot q}$ with $\sigma_{\lnot p}$. So again, without the confusion of defining new, abstract actions, our framework reproduces the intended epistemic consequences of Carl's announcement.

Finally, it is not hard to figure out how to define a dynamic model $\tilde{M}_{2}$ extending the epistemic model $M_{2}$ of Example \ref{exa:fre3}:
\begin{itemize}
\item
$f_{a}(w) = \left\{ \begin{array}{ll}
id_{\Sigma} & \textrm{if $w \models r$}\\
\Sigma \times \Sigma & \textrm{if $w \models \lnot r$,}
\end{array} \right.$
\item
$f_{b}(w) = \left\{ \begin{array}{ll}
id_{\Sigma} & \textrm{if $w \models q$}\\
\Sigma \times \Sigma & \textrm{if $w \models \lnot q$.}
\end{array} \right.$
\end{itemize}

\commentout{
We can also put our translation method to work to construct $A = \langle \Sigma, \{\approx_j | \; j \in G \}, Pre \rangle$:

\[ \Sigma = \{ \sigma_{0,p}, \sigma_{1,\neg p}, \sigma_{2,p}, \sigma_{3,\neg p} \} \]
\[ \{\approx_j | \; j \in G \} = \{ \sigma_{2,p} \approx_b \sigma_{3,\neg p} \}_{RST} \cup \{ \sigma_{0,p} \approx_a \sigma_{1,\neg p} \approx_a \sigma_{2,p} \approx_a \sigma_{3,\neg p} \}_{RST} \]
\[ Pre(\sigma_{0,p}) = p_0 \;\; Pre(\sigma_{1,\neg p}) = p_1 \;\; Pre(\sigma_{2,p}) = p_2 \;\; Pre(\sigma_{3,\neg p}) = p_3 \]

This is equivalent to the action model provided in section 2.1.
Note that since $\sigma_{2,p}$ and $\sigma_{3,\neg p}$ are indistinguishable by all agents, we can collapse them into a single node.
We cannot reduce the size of the model any further, however, because the action model framework requires that relations between actions are not world-dependent.

Now, consider the adjustment made to produce the output in figure 2.8 -- where Anne knows whether or not $p$, but is ignorant as to whether Bob speaks French.
In the dynamic framework, we need only change the static fragment of $M$ by removing edges between $w_0$ and $w_1$, and $w_2$ and $w_3$.
In the second alternative, where Anne does know whether Bob speaks French, but does not know whether or not $p$, a similar small change is required: we remove edges between $w_0$ and $w_2$, and $w_1$ and $w_3$.
In neither case are changes to the indistinguishability functions required, because the relevant \textit{actions} remain the same: note, however, that the translations of these dynamic models will result in action models with altered indistinguishability relations.
In this sense, the dynamic framework succeeds to be modular in a way which the action model framework is not.
}
%
%
%
%

% bare bones proof of axiomatization relative to W, \Sigma, f.
% relegated complexity measure details to appendix to save space, provide axiom system and translation here.

%adam2: I'm incorportating the appendix into this section (which used to be a subsection), since we have a higher page limit now.
\section{Axiomatization}

%adam2
%We relegate the mathematical details to an appendix and here provide a high-level overview of the technical results.
Our approach here follows much of the literature in seeking an axiomatization by way of \emph{reduction schemes} that effectively transform statements in $\mathcal{L}_{DL}$ into equivalent statements in some other language, and then axiomatizing that other language. Unlike PAL, however, it is easy to show that we cannot hope to reduce $\mathcal{L}_{DL}$ to the basic epistemic language $\mathcal{L}_{EL}$---it is easy to produce pairs of dynamic models whose epistemic parts are bisimilar but which satisfy different formulas of $\mathcal{L}_{DL}$ at bisimilar worlds.
%in two steps: (i) we provide an axiomatization for $\mathcal{L}_{EL^+}$ and prove completeness, and (ii) we provide an axiomatization for $\mathcal{L}_{DL}$ and prove that it reduces to the one provided in (i).

Let $\mathcal{L}_{DL}^{+}$ and $\mathcal{L}_{EL}^{+}$ denote the languages $\mathcal{L}_{DL}$ and $\mathcal{L}_{EL}$, respectively, augmented with additional primitive formulas $\xi_{j,\sigma,\sigma'}$ for each $j \in G$ and $\sigma, \sigma' \in \Sigma$. Interpret these new formulas in dynamic models as follows:
$$
\begin{array}{lcl}
(M,w) \models \xi_{j,\sigma,\sigma'} & \textrm{ iff } & (\sigma,\sigma') \in f_j(w).
\end{array}
$$
Intuitively, $\xi_{j,\sigma,\sigma'}$ says that if action $\sigma$ is performed, agent $j$ cannot rule out that $\sigma'$ was the action performed. Thus, $\mathcal{L}_{EL}^{+}$ can talk about both knowledge of the agents and action indistinguishability.

Somewhat surprisingly, $\mathcal{L}_{DL}$ is reducible to $\mathcal{L}_{EL}^{+}$: every formula of $\mathcal{L}_{DL}$ is equivalent to a formula in $\mathcal{L}_{EL}^{+}$, and this equivalence can be captured by reduction schemes that allow us to provide a sound and complete axiomatization of $\mathcal{L}_{DL}^{+}$ with respect to the class of all dynamic models that satisfy (C1) and (C2). It can be shown, however, that $\mathcal{L}_{DL}^{+}$ is strictly more expressive than $\mathcal{L}_{DL}$, so this result leaves something to be desired. This is the subject of ongoing research.

\subsection{Soundness and completeness of $\mathcal{L}_{EL}^{+}$}

Fix a finite set of actions $\Sigma$ together with a precondition function $Pre: \Sigma \to \mathcal{L}_{EL}$. We begin by axiomatizing $\mathcal{L}_{EL}^{+}$ and then turn to $\mathcal{L}_{DL}^{+}$. Consider the following axioms:

\begin{enumerate}
	\item All instantiations of propositional tautologies
	\item $K_a(\phi \rightarrow \psi) \rightarrow (K_a\phi \rightarrow K_a\psi)$
	\item $K_a\phi \rightarrow \phi$
	\item $K_a\phi \rightarrow K_aK_a\phi$
	\item $\neg K_a \phi \rightarrow K_a \neg K_a \phi$
	\item $\xi$ axioms:
	\begin{enumerate}
		\item $\xi_{j,\sigma,\sigma}$
		\item $\xi_{j,\sigma,\sigma'} \rightarrow \xi_{j,\sigma',\sigma}$
		\item $\xi_{j,\sigma,\sigma'} \rightarrow (\xi_{j,\sigma',\sigma''} \rightarrow \xi_{j,\sigma,\sigma''})$
	\end{enumerate}
	\item Interaction axiom:%
	%adam1: added footnote
	\footnote{In fact, the scheme $\lnot \xi_{j,\sigma,\sigma'} \rightarrow K_j \lnot \xi_{j,\sigma,\sigma'}$ is also valid, but it can be proved from 7(a) together with the $\mathsf{S5}$ axioms for $K_{j}$.}
	\begin{enumerate}
		\item $\xi_{j,\sigma,\sigma'} \rightarrow K_j \xi_{j,\sigma,\sigma'}$
	\end{enumerate}
	\item From $\phi$ and $\phi \rightarrow \psi$, infer $\psi$.
	\item From $\phi$, infer $K_j \phi$. 

\end{enumerate}

Soundness is easy,
while completeness can be established by a fairly standard canonical model construction.
%adam2: removed Pre, added V^c
Let $M^c = \langle W^c, \{ \sim_j^c \: : \: j \in G \}, \{ f_j^c \: : \: j\in G \}, V^c \rangle$ be defined as follows:

\begin{itemize}
	\item $M^c = \{ \Gamma \subseteq \mathcal{L}_{EL}^{+} \: : \: \Gamma \text{ is maximally consistent} \}$
	\item $\Gamma \sim^c_j \Delta$ iff $ \{ K_j \phi \: : \: K_j \phi \in \Gamma \} = \{ K_j \phi \: : \: K_j \phi \in \Delta \}$
	\item $(\sigma,\sigma')\in f^c_j(\Gamma)$ iff $\xi_{j,\sigma,\sigma'} \in \Gamma$
	%adam2: added
	\item $\Gamma \in V^c(p)$ iff $p \in \Gamma$.
\end{itemize}

We must show that $M^c$ is a dynamic model
%adam2: added
over $(\Sigma, Pre)$.
That $\sim^c_j$ is an equivalence relation follows from standard proofs of the completeness of $\mathsf{S5}$.
We show that $f^c_j$ satisfies conditions (C1) and (C2).

(C1)
Suppose that $\Gamma \sim^c_j \Delta$; we show that $f^c_j(\Gamma)=f^c_j(\Delta)$.
Since $\sim^c_j$ is an equivalence relation, suppose without loss of generality that $(\sigma,\sigma')\in f^c_j(\Gamma)$. 
Then $\xi_{j,\sigma,\sigma'} \in \Gamma$.
By axiom 7a, $K_j \xi_{j,\sigma,\sigma'} \in \Gamma$; by the way that $\sim_j^c$ is defined, we know then that $K_j \xi_{j,\sigma,\sigma'} \in \Delta$.
Finally, by axiom 3, we have that $\xi_{j,\sigma,\sigma'} \in \Delta$, which gives the desired result that $(\sigma,\sigma')\in f^c_j(\Delta)$.

(C2) follows easily from the $\xi$ axioms.

The equivalence
$$(M^c, \Gamma) \models \phi \text{ iff } \phi \in \Gamma$$
(i.e., the Truth Lemma) is proved in the standard way by structural induction.

\subsection{Soundness and completeness of $\mathcal{L}_{DL}^{+}$}
%will1: change necessary for translation to work. this is ok since we fixed Pre.
%Let $pre_\sigma$ abbreviate the formula $\lnot[\sigma]\falsum$
Let $pre_\sigma$ abbreviate the epistemic formula $Pre(\sigma)$, and consider the following additional axioms:

\begin{enumerate}
	%\item All instantiations of propositional tautologies
	%\item $K_a(\phi \rightarrow \psi) \rightarrow (K_a\phi \rightarrow K_a\psi)$
	%\item $K_a\phi \rightarrow \phi$
	%\item $K_a\phi \rightarrow K_aK_a\phi$
	%\item $\neg K_a \phi \rightarrow K_a \neg K_a \phi$
	\item[10.] Action axioms:
		\begin{enumerate}
		%will1: additional axiom needed for completeness proof
		\item $[\sigma](\phi \rightarrow \psi) \rightarrow ([\sigma]\phi \rightarrow [\sigma]\psi)$
		\item $[\sigma]p \leftrightarrow (pre_\sigma \rightarrow p)$
		\item $[\sigma]\neg \phi \leftrightarrow (pre_\sigma \rightarrow \neg [\sigma]\phi)$
		\item $[\sigma](\phi \wedge \psi) \leftrightarrow ([\sigma]\phi \wedge [\sigma]\psi)$
		%% we need \Sigma to be finite.
		\item $[\sigma]K_a \phi \leftrightarrow \big(pre_\sigma \rightarrow \displaystyle \bigwedge_{\sigma' \in \Sigma} (\xi_{a,\sigma,\sigma'} \rightarrow K_a[\sigma'] \phi)\big)$		
		%adam2: this is not supposed to be here, right?
		%\item $[\sigma][\sigma']\phi \leftrightarrow (pre_\sigma \rightarrow [\sigma']\phi)$
		\end{enumerate}
	%\item Inference rule:
	%\begin{enumerate}
		%\item From $\phi$ and $\phi \rightarrow \psi$, infer $\psi$.
		%\item From $\phi$, infer $K_a\phi$.
		\item[11.] From $\phi$, infer $[\sigma]\phi$.
	%\end{enumerate}

\end{enumerate}

The soundness of most of the above axioms is immediate; we show the soundness of (10e).

($\Rightarrow$)
Suppose that $(M,w) \models [\sigma] K_a \phi$ and assume that $(M,w) \models pre_\sigma$ (otherwise the equivalence is trivial). 
Let $\sigma' \in \Sigma$ be such that $w \vDash \xi_{a,\sigma,\sigma'}$ and consider any world $w'$ with $w \sim_{a} w'$.
We wish to show that $w' \vDash [\sigma'] \phi$.
We know that $(M^{+},(w,\sigma))\vDash K_a \phi$.
Since $w \vDash \xi_{a,\sigma,\sigma'}$, we also know that $(\sigma,\sigma')\in f_a(w)$,
so $(w,\sigma) \sim_a^{+} (w',\sigma')$;
it follows that $(w',\sigma') \vDash \phi$, which means that $w'\vDash [\sigma']\phi$.

($\Leftarrow$)
Suppose that $w \vDash pre_\sigma \rightarrow \bigwedge_{\sigma' \in \Sigma} (\xi_{a,\sigma,\sigma'} \rightarrow K_a[\sigma'] \phi)$, and that $w\vDash pre_\sigma$.
We wish to show that $(w,\sigma)\vDash K_a \phi$.
Consider any $(w',\sigma') \sim_a^{+} (w,\sigma)$; we will show that $(w',\sigma') \vDash \phi$.
Since $(\sigma,\sigma')\in f_a(w)$, we know that $w\vDash \xi_{a,\sigma,\sigma'}$, 
so by assumption, $w \vDash K_a[\sigma'] \phi$.
Then since $w\sim_a w'$, we have that $w'\vDash [\sigma'] \phi$, which implies that $(w',\sigma')\vDash \phi$, as desired.

%We show the completeness of this axiomatization by closely imitating the proof in \cite[\S 7.6]{vDvdHK08}.
%We can view the above ``Action axioms'' as \textit{reduction schemes}, allowing any formula in $\mathcal{L}_{DL}^{+}$ to be translated to an equivalent formula in $\mathcal{L}_{EL}^{+}$.
%Since $\mathcal{L}_{EL}^{+}$ has a complete axiomatization, this completes the proof.

For completeness, consider the following translation:

\begin{eqnarray*}
t(p) & = & p\\
t(\xi_{j,\sigma,\sigma'}) & = & \xi_{j,\sigma,\sigma'}\\
t(\neg \phi) & = & \neg t(\phi)\\
t(\phi \wedge \psi) & = & t(\phi) \wedge t(\psi)\\
t(K_a \phi ) & = & K_a t(\phi)\\
t([\sigma] p ) & = & pre_\sigma \rightarrow p\\
t([\sigma] \neg \phi) & = & pre_\sigma \rightarrow \neg t([\sigma] \phi)\\
t([\sigma] (\phi \wedge \psi)) & = & t([\sigma] \phi) \wedge t([\sigma] \psi)\\
t([\sigma] K_a \phi) & = & pre_\sigma \rightarrow \displaystyle \bigwedge_{\sigma' \in \Sigma} (\xi_{a,\sigma,\sigma'} \rightarrow K_a \, t([\sigma'] \phi))\\
%will1: fixed this part of the translation by enclosing formula in t()
t([\sigma][\sigma']\phi) & = & t([\sigma]t([\sigma']\phi))
\end{eqnarray*}

%adam2*:
\begin{proposition}
For all formulas $\phi \in \mathcal{L}_{DL}^{+}$, $t(\phi)$ is provably equivalent to $\phi$ and $t(\phi) \in \mathcal{L}_{EL}^{+}$.
\end{proposition}

\begin{proof}
We proceed by induction on the \emph{action nesting depth} of $\phi$, defined in the obvious way:
\begin{eqnarray*}
d(p) & = & 0\\
d(\xi_{j,\sigma,\sigma'}) & = & 0\\
d(\neg \phi) & = & d(\phi)\\
d(\phi \wedge \psi) & = & \max(d(\phi),d(\psi))\\
d(K_j \phi) & = & d(\phi)\\
d([\sigma]\phi) & = & d(\phi) + 1.
\end{eqnarray*}
The case $d(\phi)=0$ is immediate. So suppose the result holds for all formulas with nesting depth less than $n$, and let $\phi \in \mathcal{L}_{DL}^{+}$ be such that $d(\phi) \leq n$.

We now proceed via a subinduction on the \textit{weight} of $\phi$, defined as follows:
\begin{eqnarray*}
w(p) & = & 1\\
w(\xi_{j,\sigma,\sigma'}) & = & 1\\
w(\neg \phi) & = & w(\phi) + 1\\
w(\phi \wedge \psi) & = & max(w(\phi),w(\psi)) + 1\\
w(K_j \phi) & = & w(\phi) + 1\\
%will1: slight adjustment to weight here (10 now, instead of 6)
w([\sigma]\phi) & = & w(\phi) + 1
\end{eqnarray*}
The base case where $w(\phi) = 1$ is again immediate. So suppose inductively the result holds for formulas of weight less than $w(\phi)$. The proof now breaks into cases depending on the structure of $\phi$, since this determines which recursive clause of the definition of $t$ is relevant. The inducutive steps corresponding to the Boolean connectives and the $K_{a}$ modalities are straightforward, so we move to the case where $\phi = [\sigma] \psi$; this in turn naturally breaks into several subcases depending on the structure of $\psi$:
\begin{itemize}
\item
If $\psi = p$, then $t(\phi) = t([\sigma]p) = pre_\sigma \lthen p$ and we are done by axiom (10b).
\item
If $\psi = \lnot \chi$, then $t(\phi) = t([\sigma]\lnot \chi) = pre_\sigma \lthen \lnot t([\sigma]\chi)$. Clearly $w([\sigma]\chi) < w(\phi)$, so by the inductive hypothesis we know that $t([\sigma]\chi) \in \mathcal{L}_{EL}^{+}$ and is provably equivalent to $[\sigma]\chi$. It follows immediately that $t(\phi) \in \mathcal{L}_{EL}^{+}$ and, by axiom (10c), that $t(\phi)$ is provably equivalent to $\phi$.
\item
If $\psi = \chi_1 \land \chi_2$, then $t(\phi) = t([\sigma]\chi_1) \land t([\sigma]\chi_2)$. Clearly $w([\sigma]\chi_1) < w(\phi)$ and $w([\sigma]\chi_2) < w(\phi)$, so by the inductive hypothesis we know that $t([\sigma]\chi_1) \in \mathcal{L}_{EL}^{+}$ and $t([\sigma]\chi_2) \in \mathcal{L}_{EL}^{+}$, and they are provably equivalent to $[\sigma]\chi_1$ and $[\sigma]\chi_2$, respectively. It follows immediately that $t(\phi) \in \mathcal{L}_{EL}^{+}$ and, by axiom (10d), that $t(\phi)$ is provably equivalent to $\phi$.
\item
If $\psi = K_{j} \chi$, then $t(\phi) = t([\sigma]K_{j} \chi) = pre_\sigma \rightarrow \bigwedge_{\sigma' \in \Sigma} (\xi_{j,\sigma,\sigma'} \rightarrow K_j \, t([\sigma'] \chi))$. Clearly, for each $\sigma' \in \Sigma$, $w([\sigma']\chi) < w(\phi)$, so by the inductive hypothesis we know that $t([\sigma']\chi) \in \mathcal{L}_{EL}^{+}$ and is provably equivalent to $[\sigma']\chi$. It follows immediately that $t(\phi) \in \mathcal{L}_{EL}^{+}$ and, by axiom (10e), that $t(\phi)$ is provably equivalent to $\phi$, as desired.
\item
Finally, if $\psi = [\sigma']\chi$, then $t(\phi) = t([\sigma][\sigma']\chi) = t([\sigma]t([\sigma']\chi))$. Let $\tilde{\psi} = t([\sigma']\chi)$; clearly $w([\sigma']\chi) < w(\phi)$, so by the inductive hypothesis we know that $\tilde{\psi} = t([\sigma']\chi) \in \mathcal{L}_{EL}^{+}$ and $\tilde{\psi}$ is provably equivalent to $[\sigma']\chi$. So we have $t(\phi) = t([\sigma]t([\sigma']\chi)) = t([\sigma]\tilde{\psi})$. Now the weight of $\tilde{\psi}$ may be very large, so we can't apply our inner inductive hypothesis again here. However, since $\tilde{\psi} \in \mathcal{L}_{EL}^{+}$, it is easy to see that $d([\sigma]\tilde{\psi}) = 1 < d(\phi)$, so we can appeal to our \textit{outer} inductive hypothesis to conclude that $t(\phi) = t([\sigma]\tilde{\psi}) \in \mathcal{L}_{EL}^{+}$ and is provably equivalent to $[\sigma]\tilde{\psi}$. Moreover, since $\tilde{\psi}$ is provably equivalent to $[\sigma']\chi$, using axiom (10a) it is easy to show that $[\sigma]\tilde{\psi}$ is provably equivalent to $[\sigma][\sigma']\chi$, whence $t(\phi)$ is provably equivalent to $\phi$, as desired. \qedhere
\end{itemize}
\end{proof}

Completeness is an immediate corollary.

\commentout{
This translation produces provably equivalent formulas according to the axiom system above.
We show this by assigning each formula a weight, and inducting over those weights.

We need to show that the weight is reduced or not increased in all steps in the translation. For most cases, it is obvious. Consider the case for $[\sigma]K_j\phi$:

%will1: corrected mistake in calculation
\begin{eqnarray*}
w([\sigma]K_j \phi) & = & (10 + w(\sigma))*(1 + w(\phi))\\
& = & 10 + w(\sigma) + 10*w(\phi) + w(\sigma)w(\phi)\\
& > & |\Sigma| + 8 + (10 + w(\sigma))*w(\phi)\\
& \geq & w(pre_\sigma \rightarrow \bigwedge\limits_{\sigma' \in \Sigma} \xi_{(a,\sigma,\sigma')} \rightarrow K_a(pre_{\sigma'} \rightarrow [\sigma] \phi ).
\end{eqnarray*}

Note that in the case of $[\sigma][\sigma']\phi$, the translation does not \textit{reduce} the weight of the formula, but keeps it the same. 
Since there may only be finitely many iterations of the action operator, this is not a problem; the translation will `bottom out' at some point.

}

%bibliography
\bibliographystyle{eptcs}
\bibliography{abjorndahl}

\commentout{

}

\end{document}